\documentclass[preprint,11pt]{elsarticle}

\usepackage{lineno,hyperref}
\usepackage{amsmath}
\usepackage{graphicx}
\usepackage{enumerate}
\usepackage{url}
\usepackage{color}
\usepackage[utf8]{inputenc}
\usepackage{amsthm} 
\usepackage{amssymb}
\usepackage{multirow}
\usepackage{latexsym}
\usepackage{enumitem}
\usepackage{float}
\usepackage{setspace}

\usepackage{dsfont}

\journal{Journal of \LaTeX\ Templates}

\newcommand{\PP}{\mathbb{P}\/}
\newcommand{\N}{\mathbb{N}}
\newcommand{\R}{\mathbb{R}}

\newcommand{\EE}{\mathbb{E}\,}

\def\id{\hbox{{\rm\textbf 1}\kern-.5em\hbox{{\rm\textbf 1}}}\! }

\newcommand{\bi}{\begin{itemize}}
\newcommand{\ei}{\end{itemize}}
\newcommand{\beq}{\begin{equation}}
\newcommand{\eeq}{\end{equation}}

\newtheorem{theorem}{Theorem}
\newtheorem{definition}{Definition}
\newtheorem{proposition}{Proposition}
\newtheorem{lemma}{Lemma}
\newtheorem{remark}{Remark}
\newtheorem{example}{Example}
\newcounter {num} 

\newcounter {no} 

\newcounter {nrq} 

\begin{document}

\begin{frontmatter}
\title{Markov-Based Modelling for Reservoir Management: Assessing Reliability and Resilience}

\author{M.L. G\'amiz\corref{mycorrespondingauthor}\corref{mymainaddress1}}
\ead{mgamiz@ugr.es}
\cortext[mymainaddress1]{Universidad de Granada (SPAIN)}

\author{N. Limnios\corref{mymainaddress2}}
\ead{nlimnios@utc.fr}
\cortext[mymainaddress2]{Universit\'e de Technologie de Compi\`egne (FRANCE)}

\author{D. Montoro-Cazorla\corref{mymainaddress3}}
\ead{dmontoro@ujaen.es}
\cortext[mymainaddress3]{University of Jaén (SPAIN)}

\author{M.C. Segovia-Garc\'ia\corref{mymainaddress1}}
\ead{msegovia@ugr.es}

\begin{abstract}

This paper develops a comprehensive Markov-based framework for modelling reservoir behaviour and assessing key performance measures such as reliability and resilience. We first formulate a stochastic model for a finite-capacity dam, analysing its long-term storage dynamics under both independent and identically distributed inflows, following the Moran model, and correlated inflows represented by an ergodic Markov chain in the Lloyd formulation. For this finite case, we establish stationary water balance relations and derive asymptotic results, including a central limit theorem for storage levels. The analysis is then extended to an infinite-capacity reservoir, for which normal limit distributions and analogous long-term properties are obtained. A continuous-state formulation is also introduced to represent reservoirs with continuous inflow processes, generalizing the discrete-state framework. On this basis, we define and evaluate reliability and resilience metrics within the proposed Markovian context. The applicability of the methodology is demonstrated through a real-world case study of the Quiebrajano dam, illustrating how the developed models can support efficient and sustainable reservoir management under hydrological uncertainty.

\end{abstract}


\end{frontmatter}


\section{Introduction}\label{sec1}
Efficient water management is a critical challenge in the context of climate change, population growth, and increasing demand for water resources. Hydrometeorological variability introduces substantial uncertainty in reservoir inflows and outflows, complicating planning and operations and affecting water availability for domestic consumption, irrigation, energy generation, and industry \cite{Didier2018}, \cite{Ren2020}. In multi-purpose reservoirs, the need to balance competing water uses adds further complexity, requiring operational strategies that maximize resource efficiency and ensure availability under changing conditions \cite{Allawi2019}, \cite{Tukimat2014}.  

To address these uncertainties, stochastic models are commonly employed in the analysis and optimization of water systems. Such models capture the inherent variability of hydrological processes and provide valuable insights for decision-making \cite{Asefa2014}, \cite{Zhao2024}. Among these approaches, \textit{Markov chains} have proven particularly effective for representing the stochastic nature of inflow dynamics and reservoir storage behaviour. Previous studies have applied Markov models to analyse precipitation occurrence and transitions between dry and wet periods \cite{Feyerherm1967}, \cite{Gates1976}, \cite{Martin1999}, \cite{Nop2021}, and to assess drought risks in rainfed agriculture using multi-state Markov chains \cite{Fadhil2021}. 
In the context of reservoir reliability, \cite{Limnios1988}, \cite{Limnios1989} employed discrete-time Markov chains based on the Moran model \cite{Moran59} to assess the likelihood of reservoir depletion or overflow under varying inflow conditions.


While reliability has long been a central concept in the analysis of engineering systems \cite{Gamiz2022}, \cite{Gamiz2023}, \cite{Gamiz2024}, it remains a rapidly evolving field within operations research and systems engineering \cite{Aven2025}. Its statistical foundations continue to develop through both parametric and nonparametric approaches \cite{Gamiz2011}. This perspective also extends to water supply networks, where statistical methods have been applied to evaluate system performance under incomplete failure data. For example, \cite{Carrion2010} examined the reliability of a water supply network using right-censored and left-truncated break data, providing a robust approach for assessing component reliability from limited observations. In recent years, \textit{resilience} has been increasingly recognized as a complementary indicator of system performance, particularly for critical infrastructures such as energy, transport, and water systems \cite{Didier2018}, \cite{Ren2020}. Resilience is broadly defined as the capacity of a system to withstand disturbances, adapt to changes, and recover from disruptions \cite{Hosseini2016}, \cite{Pawar2021}. Its evaluation commonly involves three key dimensions: \textit{absorption}, representing the system’s initial resistance to perturbations; \textit{adaptation}, describing its ability to adjust operations during a crisis; and \textit{recovery}, assessing the speed and effectiveness with which normal conditions are restored \cite{Sathurshan2022}. These aspects have been widely applied to analyse the performance of critical infrastructures exposed to extreme events, supporting strategies that enhance robustness and response capacity \cite{Diao2016}.  

Markov processes have proven effective as a framework for quantifying resilience and dynamic performance in complex systems \cite{Zeng2021}, \cite{Tan2023}, \cite{Wu2023}. They allow for the probabilistic representation of state transitions and enable the derivation of quantitative metrics that describe reliability, recovery, and overall system stability \cite{Wu2023}. Applications of Markov-based resilience assessment can be found across multiple sectors: energy systems, where they quantify failure absorption and recovery following extreme events \cite{Zeng2021}; maritime transport infrastructures, where they support the analysis of adaptive strategies against structural failures \cite{Wu2023}; and industrial systems, such as hydrogen stations and gas compressors, where they estimate recovery times and service restoration probabilities after operational disruptions \cite{He2024}, \cite{He2025}. 

In water resource systems, resilience relates not only to service continuity but also to the capacity of reservoirs to regain operational performance following droughts, floods, or other extreme events. Several studies have used Markov models to evaluate reservoir management and mitigation strategies \cite{Vogel1995}, \cite{Zolfagharpour2021}, and to derive optimal operating rules for irrigation reservoirs under uncertain inflows using stochastic control \cite{Unami2013}, although their integration into formal resilience assessment frameworks remains relatively limited.
 In modelling reservoir dynamics, the inflow process is a key stochastic component. Early approaches often assumed independent and identically distributed (i.i.d.) inflows; however, in many practical situations, inflows exhibit temporal correlation. \cite{Lloyd1963}, \cite{Lloyd1964} introduced a formulation in which reservoir inflows are modelled as a Markov chain, accounting for such correlation effects. In this paper, we adopt this framework for correlated inflows, referring to it subsequently as the \emph{Lloyd formulation} or \emph{extended Moran model} to emphasize its connection with classical discrete-time stochastic reservoir models.


Building on this background, the present study extends the use of Markov processes to the analysis of reservoir performance under uncertainty. Following the approach of \cite{Tan2023}, we define resilience measures for systems subject to disruptive events within a \textit{Markov process–based framework}. The proposed methodology provides a quantitative basis for assessing the dynamic behaviour of water systems, supporting more efficient and robust reservoir management under stochastic hydrological conditions.

This paper proposes a multi-state Markov model based on Moran model to represent the annual water balance evolution of a reservoir located in the province of Jaén (Spain). The model relies on data collected from monitoring equipment installed in the dam, which is crucial for its implementation. Two distinct models are developed based on how inflow volumes are treated: one where inflows are considered discrete variables and another where they are treated as continuous variables. For each model, analytical results describing the system asymptotic behaviour are presented, along with reliability and resilience metrics. These metrics allow for assessing the reservoir capacity to maintain safe storage levels and its recovery following events that either compromise water supply (drought) or lead to overflow (flooding). Both aspects are crucial for water governance in the studied reservoir because of their important socio-economic implications.

The remainder of this paper is organized as follows. Section \ref{sec:model} presents the general formulation of the proposed stochastic model with discrete input. Section \ref{sec:finite} develops the Markov model for a finite-capacity reservoir, addressing both the Moran case with independent inflows and the Lloyd formulation with Markov-dependent inflows, and establishes its long-term properties including a stationary water balance and a central limit theorem for storage. Section \ref{sec:infinite} extends the analysis to the infinite-capacity case, deriving analogous asymptotic results. Section \ref{sec:continuous} considers a continuous-state formulation to represent reservoirs with continuous inflow processes. Section \ref{sec:measures} defines and discusses performance indicators for reliability and resilience within the Markovian framework. Section \ref{sec:quiebrajano} illustrates the application of the proposed approach through a real case study of the Quiebrajano dam. Finally \ref{sec:conclusions} presents concluding remarks and future research.

\section{Stochastic properties of a reservoir system with discrete input}\label{sec:model}
In our practical application, to describe the dynamical evolution of a reservoir system, specifically a dam for human water supply, we will employ a 2-dimensional stochastic process  $\{(Y_n,Z_n), n\ge 0\}$ with state space $ E_y \times E_z$. In this model, $ Y_n$ will represent the total amount of rain registered in a particular area during the $n$-th year, with $n = 0, 1,  \ldots$. It is well known that the amount of rainfall in a particular area depends on various factors, including atmospheric circulation patterns, topography, proximity to bodies of water, temperature, humidity levels, elevation, land use, and, of course, climate change. While these factors contribute to the heterogeneity of rainfall, our initial approach does not explicitly incorporate them.  Instead, we assume that yearly rainfall can be represented by a sequence of independent identically distributed (i.i.d.) random variables as a first step in our analysis, $\{Y_n, n\geq 0\}$, with distribution $\{p_y; y \in E_y\}$. Next, in a more realistic scenario, we will consider $\{Y_n, n\geq 0\}$ as a stationary Markov chain with transition matrix ${\bf P}_y$.

For any $n\ge 0$, let $Z_n$ denote the volume of water at the time (year) $n$, then  $Z_0$ is the initial stored water in the reservoir. In general, we will assume that the 2-dimensional process $\{(Y,Z)\}$, is observable. This implies that, for any $n \ge 0$, we can directly observe the rainfall amount $Y_n$ as well as the stored water volume $Z_n$. Moreover, we will model the stored water on the year $n+1$ as a function of the previous values, i.e. $Z_{n+1}=f(Y_n,Z_n)$, see Figure \ref{fig:model}.

\begin{figure}[htb]
	\centering
	\vspace{-.5cm}
	\includegraphics[width=.8\textwidth]{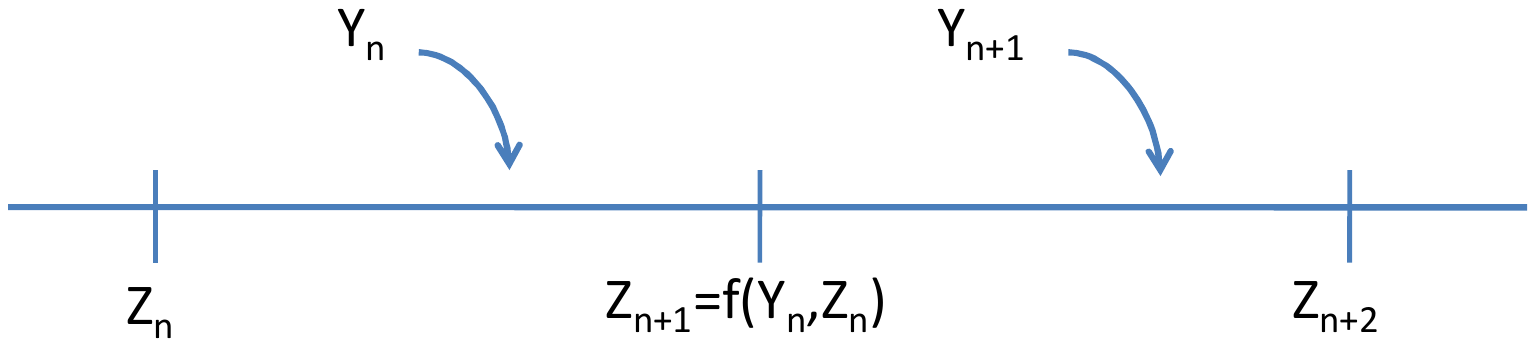}
	\vspace{-4cm}
	\caption{A  Markov model for reservoirs.}\label{fig:model}
\end{figure}

Our primary objective is to model the distribution of the stochastic process ${Z_n,, n \geq 0}$ in order to characterize the temporal dynamics of reservoir water levels. The variable $Z_n$ represents the theoretical water storage in the reservoir at the end of year $n$, derived from the hydrological balance between inflows and outflows. The actual observed volume, denoted by $\widetilde{Z}_n$, may differ from this theoretical value due to additional factors not explicitly captured in the inflow–outflow model, such as evaporation, seepage, or filtration losses. Some of these losses may be associated with structural deterioration, including the development of cracks or defects in the dam body. Therefore, analysing the discrepancy $Z_n - \widetilde{Z}_n$ provides a means to assess deviations between expected and observed performance, potentially serving as a proxy for monitoring the evolution of structural degradation over time. This relationship is illustrated in Figure~\ref{fig:water_balance}, which depicts the empirical water balance computed from the available data using the input–output equation. The curves in the figure were obtained from publicly available information on the Quiebrajano Dam, a reservoir located in southern Spain. The dotted black curve represents the daily stored water volume during 2007 ($\widetilde{Z}_n$), while the solid red curve shows the result of the daily input–output balance equation computed from the published inflow and outflow data for the same period ($Z_n$).
\begin{figure}[h]
	\centering
	\includegraphics[width=0.6\textwidth]{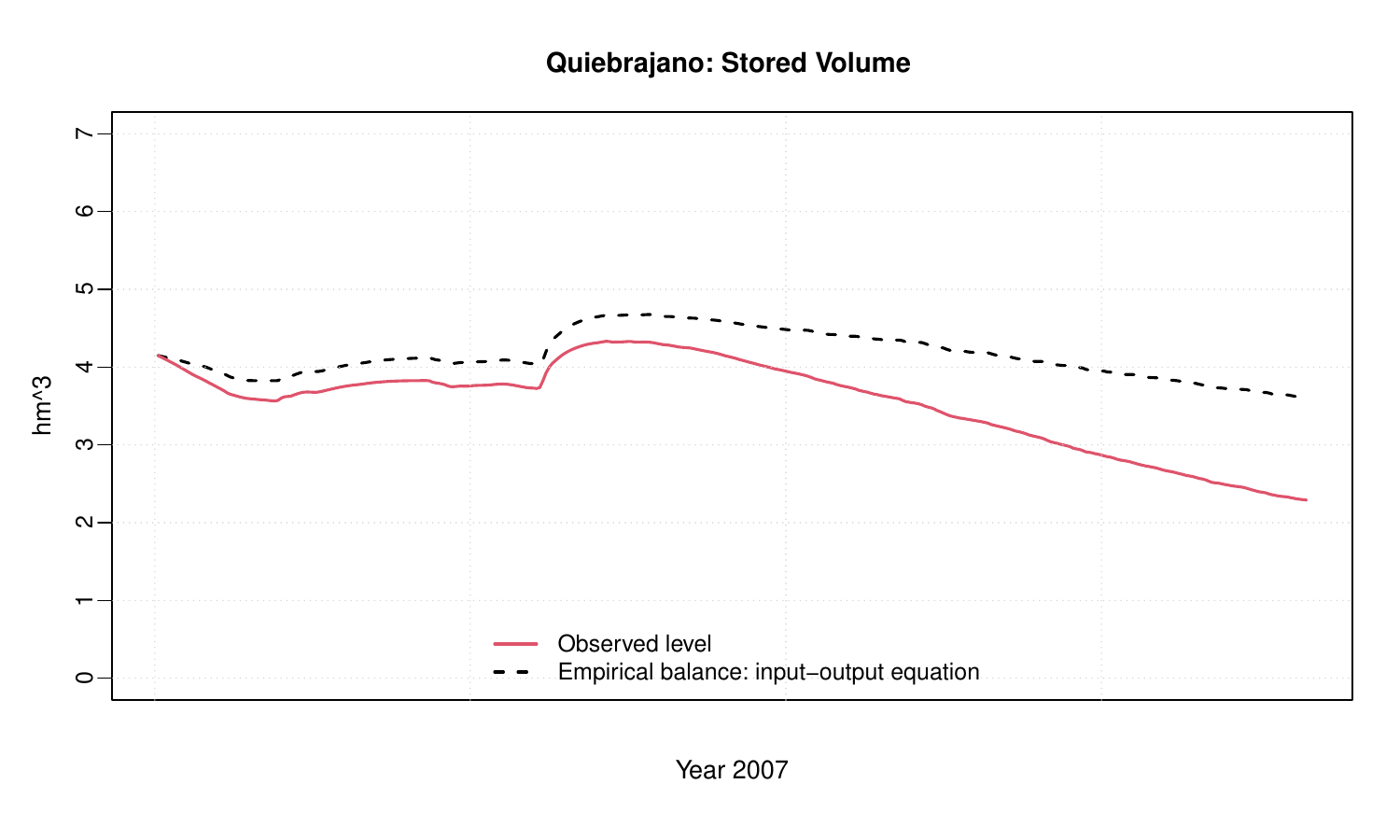}
	\vspace{-.5cm}
	\caption{Water balance during the year 2007 adjusted empirically.}\label{fig:water_balance}
\end{figure}

\section{The finite case}\label{sec:finite}
\subsection{Model description}
In our initial approach, we assume that the input sequence  $\{Y_n, n\geq 0\}$ consists of independent and identically distributed (i.i.d.) random variables taking values in a discrete set $E_y \subseteq \N$. {Let $\{p_y; y \in E_y\}$ denote the distribution of $Y_n$ for all $n$}.

We assume the Moran model \cite{Moran59}, so we have that $Z_{n}=f(Y_{n-1},Z_{n-1})$, with the following specification for the  function $f$, 
\begin{equation}\label{eq:moran}
	f(y,z)=\left\{\begin{array}{lll}
		0, & \text{if } &z+y < c_0 \\
		z+y-c_0, & \text{if } &c_0 \leq z+y \leq C_1 \\
		C_1-c_0, & \text{if } &z+y \ge C_1
	\end{array}
	\right.;
\end{equation}
where $c_0$ and $C_1$ are some positive constants, $0<c_0<C_1$. 
Compactly we can write, $Z_{n+1}=\min\{C_1-c_0,\max\{Y_n+Z_n-c_0,0\}\}$, $n \geq 0$.
For simplicity in the following and without any loss of generality we take $c_0=1$ and use notation $C$ instead of $C_1$. In these conditions, $Z$ is a Markov chain with state space $E_z=\{0,1,\ldots, C-1\}$.\\

\noindent The process $(Y,Z)$ is a Markov chain with transition probabilities as follows
\begin{eqnarray*}
	\widetilde{P}((y_0,z_0),(y_1,z_1))&=&\PP(Y_{n+1}=y_1,Z_{n+1}=z_1 \mid Y_{n}=y_0,Z_{n}=z_0 )\\
	&=&\PP(Z_{n+1}=z_1 \mid Y_{n}=y_0,Z_{n}=z_0 ) \PP(Y_{n+1}=y_1)\\
	&=& {\tt 1}_{\{z_1=(z_0+y_0-1)^+\}}p_{y_1},
\end{eqnarray*}
where $(x)^+=\max\{0,x\}$, for any $x\in \R$. We consider the lexicographical order for the elements of $E_y\times E_z$ and organize these probabilities in a matrix that we denote $\widetilde{\bf P}$. 

When the inflow process exhibits temporal dependence rather than being i.i.d., the reservoir evolution follows the so-called \emph{Lloyd formulation} of the dam equations \cite{Lloyd1964}, which generalizes Moran's original model to the case of Markov-dependent inflows.

\begin{remark}{Lloyd's formulation \cite{Lloyd1963} (extended Moran model)}\\
	When the inflow process $\{Y_n, n \ge 0\}$ is modelled as an ergodic Markov chain, with transition matrix ${\bf P}_Y$ and stationary distribution $\{p_y; \, y \in E_y\}$, the reservoir dynamics follow what \cite{Lloyd1963} referred to as the Markov inflow formulation of the dam equations. In this case, the joint process $(Y,Z)$ is itself a Markov chain with transition matrix $\widetilde{\bf P}$, whose elements are given by $
	\widetilde{P}((y_0,z_0),(y_1,z_1)) = \mathbf{1}_{\{z_1 = (z_0 + y_0 - 1)^+\}} P_{y_0,y_1}$,	for $((y_0,z_0),(y_1,z_1)) \in E = E_y \times E_z$, where the state space is ordered lexicographically.

 Let us assume that $Y_n$ has finite support $E_y=\{0,1,\ldots, K\}$. To write the transition matrix $\widetilde{\bf P}$, we consider the following notation:
\begin{itemize}
	\item $E_{(\cdot, i)}=E_y\times \{i\}$, for $i\in E_z$;
	\item ${\bf D}_{k_1:k_2}$,  is a $K$-dimensional diagonal matrix whose diagonal entries are 1 from position $k_1$ to $k_2$, and 0 elsewhere, with $1\leq k_1\leq k_2 \leq K$.
	\item ${\bf 0}_{K\times K}$ is $K$-dim square matrix of 0s;
	\item $\widetilde {\bf P}_{i,j}$ is the sub-matrix of $\widetilde{\bf P}$ with all transitions from any state of $E_{(\cdot, i)}$ to any state of $E_{(\cdot, j)}$, with $i,j=0,\ldots, C-1$.
\end{itemize} 	

We have that

\begin{equation*}
	\widetilde{\bf P}_{i,j}=\left\{
	\begin{array}{ll}
		{\bf D}_{1:2}{\bf P}_y ,             & \ i=j=0;\\
		{\bf D}_{(j-i+2):(j-i+2)}{\bf P}_y , & \ i=0, j=1,\ldots, C\!-\!2; \ \text{or}\  i\geq 0; j=i\!-\!1,\ldots, C\!-\!2 \\
		{\bf 0}_{K\times K};                 & \ i\geq 2; j=0,\ldots, i\!-\!2; \\
	    {\bf D}_{(C-1):K}{\bf P}_y;          & \ i=0,\ldots, C\!-\!1; j=C\!-\!1 \\
	\end{array}
	\right.
\end{equation*}
\end{remark}
\bigskip


\noindent Some main problems to be addressed:
\begin{enumerate}
	\item The marginal distribution of the $Z$ process, that is the distribution of the level of water stored in the reservoir
	\item To calculate the mean time to emptiness, that is, if we denote   $T_E=\inf\{n >0: Z_n < 1\}$, then we aim to calculate $MTTE=E[T_E]$;
	\item To calculate the mean time to overflow, that is, if we denote $T_O=\inf\{n >0: Z_n >C-1\}$, then we aim to calculate $MTTO=E[T_O]$.
	\item To calculate the expected value of water loss per unit time in the long run. If we denote $C_O(n)=\sum_{k=0}^n (Z_k+Y_k-C)^+$, then we would like to determine ${\lim}_{n \rightarrow \infty}\frac{ C_O(n)}{n}$.
	\item To calculate resilience metrics in order to describe how long the reservoir can operate safely and how quickly it recovers from extreme conditions
\end{enumerate}

\subsection{The long-term performance of the reservoir system}\label{algorithm}
An essential question in reservoir management is understanding the long-term behaviour of the stored volume as time progresses. Specifically, this involves determining the stationary distribution of  $\{Z_n; n\geq 0\}$, which represents the sequence of stored volumes. In this section we assume that $\{Y_n, n \ge 0\}$ is a sequence of i.i.d. random variables with distribution $\{p_y; y \in E_Y\}$, and then $\{Z_n, n\ge 0 \}$ is a Markov chain.

The stationary distribution provides insight into the probabilities of different storage levels in the reservoir over time, regardless of the initial volume, aiding in effective long-term resource management.
From the Moran model, we have that $\{Z_n; n\geq 0\}$ is an irreducible Markov chain with finite state space $E_z=\{0,1,\ldots, C-1\}$, then there exist a unique limiting distribution $\pi$, obtained as the normalized solution of the stationary equation. That is $\pi {\bf P}=\pi$, and $\pi {\bf 1}=1$, where ${\bf P}$ is the transition matrix of the Markov chain $Z$ which, according to the specifications in \eqref{eq:moran}, has the following expression
\begin{equation}\label{Pmat}
	{\bf P}=\left(
	\begin{array}{cccccc}
		p_0+p_1& p_2 & p_3& \cdots &p_{C-1}& h_C\\
		p_0& p_1 & p_2& \cdots &p_{C-2}& h_{C-1}\\
		\vdots&\vdots& \vdots &\ddots& \vdots &\vdots\\
		0 & 0 &0 &\cdots &p_0&h_1 
	\end{array} 
	\right)
\end{equation}
with $h_k=\sum_{j=k}^{+\infty}p_j$, and $p_j=\PP(Y_n=j)$, for $j=0,1, \ldots$
\begin{lemma}
	If $0<p_0$ and $p_0+p_1<1$,  then the Markov chain $\{Z_n; n\geq 0\}$ is ergodic, and its stationary distribution is given by $\pi_n= a_n \pi_0$, where \begin{equation}\label{eq:an}
	a_n=\frac{1}{p_0}\sum_{j=1}^{n}\left(\delta_{j1}-p_j\right)a_{n-j}, \ 2 \leq n \leq C-1
	\end{equation}
where $a_0=1$ and $\pi_0$ is the corresponding normalizing constant.
	
\end{lemma}
\begin{proof}
The proof is immediate from classical Markov chain arguments.

 Moreover, the  stationary distribution of $\{Z_n,n\ge 0\}$ can be built recursively as follows. For $k=1, \ldots, C-1$, $\pi_n=a_n \pi_0$,  where the coefficients $a_k$ can be constructed recursively as follows: $a_1=\frac{1-p_0-p_1}{p_0}$, and, for $n=2$, we have
	$$\pi_2=\frac{-p_2}{p_0}\pi_0+\frac{1-p_1}{p_0}\pi_1= \left(\frac{-p_2}{p_0}+\frac{1-p_1}{p_0}a_1\right)\pi_0=a_2 \pi_0.$$
and so on.  Then we can express $a_n$ as given in \eqref{eq:an}.

\end{proof}

\noindent We summarize the steps to construct the stationary distribution for the finite dam model in the following algorithm.\\

\fbox{\begin{minipage}{33em}
	\noindent{\bf Algorithm. Stationary distribution of $Z$ under the Moran model.}
		\begin{enumerate}
			\item DEFINE $a_0 := 1$.
			\item SET $n = 1$ and $a_1 = (1 - p_0 - p_1)/p_0$,
			\item SET $b_n = \sum_{j=1}^{n+1}p_{j}a_{n+1-j}$;
			\item PUT $n=n+1$ and SET $a_n=(a_{n-1}-b_{n-1})/p_0$. 
			\item If $n< C-1$, GO to step 3 else GO to step 6.
			\item SET $\pi_0=\left({\sum_{n=0}^{C-1}}a_n\right)^{-1}$, and $\pi_n= a_n \pi_0$.
		\end{enumerate}
	\end{minipage}
}

\subsection{Some limit results for the storage volume in a finite capacity dam}
In this section we study the  asymptotic properties of the process that represents the volume of water stored in the dam, $Z$. 

\begin{lemma} \label{lemma:pi}
	The 2-dimensional process $\{(Y_n,Z_n); n \ge 0\}$ is an  ergodic Markov chain with limiting distribution 
	\begin{equation}\label{eq:pi2}
		\widetilde{\pi}(y,z)=\lim_{n \rightarrow +\infty} \PP(Y_n=y,Z_n=z)=\PP(Y=y)\lim_{n \rightarrow +\infty} \PP(Z_n=z)=p_y\pi(z),
	\end{equation}

\end{lemma}

\begin{proof}The proof is straightforward since $Z_n$ is determined by  $(Y_0,\ldots, Y_{n-1})$, and $Y_n $ is independent of $ (Y_0,\ldots, Y_{n-1})$.
\end{proof}
\begin{remark}
	Given that the capacity of the dam is assumed to be finite it must happen that $\frac{Z_n}{n} \rightarrow 0$ ({\it a.s.}), as $n\rightarrow \infty$
\end{remark}
\begin{proposition}{\bf Stationary water balance (finite case)}\label{prp:as}\\
	Let $Z_n$ denote the volume of water stored in a reservoir at time $n \ge 0$. Let $Y_n$ denote the total input registered during the period $(n,n+1)$. We assume that $\{Y_n, n\ge 0\}$ form a sequence of i.i.d. random variables with common mean $\mu_y$. Let also assume the couple $(Y,Z)$ fits the Moran model, that is for any $n \ge 0$, $Z_{n+1}=\min\{\max\{0,	Z_n+Y_n-c_0\},C_1\}$. Then in stationarity $\mu_y=m_{c_0}-M$, 	where $m_{c_0}=c_0 \lim_{n\rightarrow \infty}\PP(Z_n \ge c_0)$is the long-term mean controlled outflow and $M=\EE[(Z_n+Y_n-C_1)^+]$ is the long-term mean water loss due to overflow, per unit time.
	
\end{proposition}

\begin{proof}
	Let us take $c_0=1$ and $C_1=C$, as taken above. In that case, we have that $m_{c_0}=1-\pi_0$. From Moran model, we can write
	\[
	Z_{n+1}=Z_0+\sum_{k=0}^n Y_k -\sum_{k=0}^{n} {\tt 1}_{\{Z_k >0\}}-\sum_{k=0}^n (Z_k+Y_k-C)^+,
	\]
	then, 
	\[
	\frac{Z_{n+1}}{n+1}=\frac{Z_{0}}{n+1}+\frac{1}{n+1}\sum_{k=0}^n Y_k -\frac{1}{n+1}\sum_{k=0}^{n} {\tt 1}_{\{Z_k >0\}}-\frac{1}{n+1}\sum_{k=0}^n (Z_k+Y_k-C)^+,
	\]
	As $n\rightarrow +\infty$, it is clear that $(n+1)^{-1}Z_{0} \rightarrow 0$, and using the strong law of large numbers for i.i.d. variables, we have $\frac{1}{n+1}\sum_{k=0}^n Y_k \stackrel{a.s.}{\longrightarrow} \mu_y$,	and, $
	\frac{1}{n+1}\sum_{k=0}^{n} {\bf 1}_{\{Z_k >0\}}  \stackrel{a.s.}{\longrightarrow} 1-\pi_0$ 
	and, finally $
	\frac{1}{n+1}\sum_{k=0}^n (Z_k+Y_k-C)^+ \stackrel{a.s.}{\longrightarrow} \sum_{y,z}\widetilde{\pi}(y,z)(z+y-C)^+$.

	Let us define $M:=\sum_{y,z}\widetilde{\pi}(y,z)(z+y-C)^+=\sum_{y,z}\pi(z)p_y(z+y-C)^+$, then we get that 
	
	\[
	\frac{Z_n}{n} \stackrel{a.s.}{\longrightarrow} \mu_y -(1-\pi_0)-M, \ \ \ {\rm as \ } n \rightarrow +\infty,
	\]
since, it has to be that $Z_n/n \rightarrow 0$, then stationary water balance equation holds.
\end{proof}
\begin{remark}
{If $\{Y_n, n \geq 0\}$ is a stationary MC, then we can apply the Law of Large Numbers for Markov Chains and the result of Proposition \ref{prp:as} still holds.}
\end{remark}

\begin{proposition}{\bf Strong law for the storage volume (finite case)}\\
	Let $Z_k$ denote the volume of water stored at time $k\ge 0$, then $
	\frac{1}{n+1}\sum_{k=0}^n Z_k \rightarrow \mu_z$, 
where $\mu_z=\sum_{i=0}^{C-1} z_i \ \pi_i$, with $\pi_i= \lim_{n\rightarrow \infty}\PP(Z_n=z_i)$, $i=0,\ldots, C-1$.
\end{proposition}
\begin{proof}
	The result follows immediately from the ergodicity of the  Markov chain $Z$.
\end{proof}

\begin{theorem} {\bf A Central limit theorem for storage volume (finite case).}\\
Let $Z_k$ denote the volume of water stored in a reservoir at time $k \ge 0$, and define $S_n=\sum_{k=0}^{n-1}Z_k$, the cumulative storage volume up to time $n$.

Then 
\[
\sqrt{n}\left(S_n- n\mu_z\right) \stackrel{\cal D}{\longrightarrow} N(0, \sigma^2) \ \ {\rm as \ } n \rightarrow +\infty,
\]
with
\begin{equation}\label{eq:sigma}
\sigma^2=\widetilde{\bf g}_z  \left\{ {\bf D}_{\pi}(2({\bf I}- {\bf P}+{\boldsymbol \Pi})^{-1}-{\bf I })\right\}\widetilde{\bf g}_z^{\top},
\end{equation}
where we use the following notation: ${\bf I}$ is the identity matrix; ${\bf D}_{\pi}={\tt diag}(\pi_i; i=0,\ldots, C-1)$; ${\boldsymbol  \Pi} = {\bf 1} \pi$, with ${\bf 1}=(1,1,1,\ldots )^{\top}$; and vector $\widetilde{\bf g}_z=(z_i-\mu_z;i=0,\ldots, C-1)$.
\end{theorem}
\begin{proof}
	
To prove this result we consider the works on the central limit theorem for finite recurrent Markov chains presented in \cite{Port1994} and \cite{TrevezasLimnios09}. To be concrete, in \cite{TrevezasLimnios09} the associated variance on the CLT is presented in different matrix forms which result from different proofs of the CLT that can be found already in the literature, the authors focus their analysis is the form for the variance given in \cite{Port1994}.

We follow the terminology in \cite{TrevezasLimnios09} and define the following:
\begin{itemize}
	\item $Z$ is a Markov chain on a finite state space $ E_z=\{0,1,\ldots, C-1\}$, with transition matrix ${\bf P}$ and stationary distribution ${\pi}$ defined above;
	\item $g : E_z \rightarrow \R$, where $g(i)=i$, for all $i \in E_z$;
	\item $\mu_g:=\EE_{\widetilde{\pi}}[g(Z)]=\mu_z$;
	\item $\widetilde{g}(i):=g(i)-\mu_g$; 
	\item $S_n:=\sum_{k=0}^{n-1} g(Z_k)$.
\end{itemize}
Then, in \cite{Port1994} it is proven that $
\sqrt{n}\left(\frac{S_n}{n}-\mu_g\right) \stackrel{D}{\longrightarrow} N (0, \sigma^2)$, 
and following \cite{TrevezasLimnios09} it can be obtained the  matrix form for the variance specified in \eqref{eq:sigma}.
\end{proof}

\section{The semi-infinite case}\label{sec:infinite}
In this case we assume that the reservoir has infinity capacity $C=\infty$, then the volume stored at time $n \ge 1$ can be obtained using the input sequence $\{Y_n, n\ge 0\}$ as
 \begin{equation}\label{eq:moran_infinite}
 Z_n=\max\{Z_{n-1}+Y_{n-1}-c_0,0\},
 \end{equation}
  $n>0$. As in the previous section, let us denote $p_i=\PP(Y_n=i)$, for $i \in \N$, the distribution low of the variable input $Y_n$, for all $n\ge 0$. In practice it is not overly restrictive to assume that the sequence $\{p_i, i \in \N\}$ has finite support; that is, there exists some $K \in \N$, which can be arbitrarily large, such that $p_{K+i}=0$ for $i \ge 1$. Then $Z_n$ is a Markov chain with infinite transition matrix given by
\begin{equation}\label{Pmat2}
	{\bf P}=\left(
	\begin{array}{ccccccccccc}
		p_0+p_1& p_2 & p_3&\cdots& p_{K-1}& p_K  &0 &0 &0&\cdots \\
		p_0& p_1 & p_2&p_3 &\cdots & p_{K-1}& p_K & 0&0&\cdots\\
		0& p_0 & p_1& p_2 &p_3 &\cdots& p_{K-1}&p_K&0&\cdots\\
		\vdots&\vdots& \vdots &\vdots& \vdots & \vdots & \vdots & \vdots &\vdots&\ddots\\
		\end{array} 
	\right)
\end{equation}

\begin{proposition}\label{prp:pi_infinite} \textbf{The stationary distribution (semi-infinite case).} \\
	Let $\{Z_n; n\geq 0\}$ be a DTMC with transition matrix  ${\bf P}$ as in \eqref{Pmat2}, with $0<p_0<1$ and $p_0+p_1<1$, and 
	\begin{equation}\label{eq:cond}
	\sum_{k=1}^{K-1}\sum_{j=k+1}^{K} \frac{p_j}{p_0} <1.
\end{equation}Then there exist a stationary distribution $\pi=(\pi_0,\pi_1, \ldots)$ such that $\pi_n=\pi_0 a_n$, with $a_n=a_n(p_0,p_1, \ldots)$ is a function of the probability distribution of $Y$.

\end{proposition}

\begin{proof}
	For the case of a semi-infinite dam (see \cite{Yeo1961}), we assume that the state space is $E_z=\N$. Let us assume that the probability generating function (PGF) of the stationary distribution $\pi$ exists, and let us denote $G(s)=\sum_{i=0}^{+\infty} \pi_i s^i$, for $|s|\leq 1$.  Also, let us denote  $H(s)=\sum_{i=0}^{K} p_i s^i$, for $|s|\leq1$, the PGF of the distribution of the input $Y$. Given $s>0$, the system of linear equations resulting from  $\pi {\bf P}=\pi$ can be transformed as follows
	\[
	\begin{array}{rrrrrrr}
		\pi_0 s  \ =& sp_0+& p_1s  \  \pi_0 + &  p_0  \   \pi_1s+ & & & \\ 
		\pi_1 s^2 =&+&  p_2s^2  \pi_0+ &  p_1 s   \pi_1s +& p_0 \ \pi_2s^2+ &&\\
		\pi_2 s^3  =&+&  p_3 s^3 \pi_0+ &  p_2 s^2 \pi_1s+ & p_1 s \pi_2  s^2+& p_0 \ \pi_3s^3+& \\
		\pi_3 s^4=&+& p_4s^4  \pi_0+ & p_3 s^3 \pi_1 s+& p_2 s^2 \pi_2 s^2+ & p_1 s \pi_3 s^3+&p_0 \ \pi_4s^4  \\
		\multicolumn{6}{r}{\cdots\cdots\cdots\cdots\cdots\cdots\cdots\cdots\cdots\cdots\cdots\cdots\cdots\cdots\cdots\cdots\cdots\cdots} \\
	\end{array}
	\]
	which can be written in terms of $G$ and $H$ as $sG(s) - H(s)G(s)= (s-1) p_0 \pi_0$, 
	or, equivalently $G(s)(H(s)-s)=(1-s)p_0\pi_0$.\\
	Since $H(1)=1$, we can divide both sides of the above expression by the binomial $1-s$ so we get that 
	\begin{equation}\label{eq:G}
	G(s)= \frac{\pi_0}{1- \sum_{k=1}^{K-1}q_ks^k},
	\end{equation}
	with $q_k=\sum_{j=k+1}^{K} \frac{p_j}{p_0}$, for $k=1,\ldots, K-1$.

	We can expand in a power series the right hand term of expression \eqref{eq:G} to obtain closed form of expressions for the probabilities $\pi_i$, for $i =0, 1, \ldots$. Let us denote $B(s)= \sum_{k=1}^{K-1}q_ks^k$. Under \eqref{eq:cond}, we have that $| B(s)| <1$, for $|s| \leq 1$, then  
	\begin{equation}\label{eq:B}
	G(s) = \pi_0 \sum_{n=0}^{\infty}(B(s))^n < \infty.
	\end{equation}

	The $n$th term of the sum in \eqref{eq:B} involves the $n$th power of a polynomial of degree $K-1$, $B(s)$. Then we have that 
	
	\begin{equation}\label{eq:G2}
		G(s)=\pi_0 \sum_{n=0}^{+\infty} 
		\sum_{k_1=1}^{K-1}\sum_{k_2=1}^{K-1}\cdots \sum_{k_n=1}^{K-1}q_{k_1}q_{k_2}\cdots q_{k_n}s^{k_1+k_2+\ldots+k_n}.
	\end{equation}
		We can reorder the terms in \eqref{eq:G2}, to  obtain an expression for $G(s)$ as a power series such as $	G(s)=\pi_0\sum_{n=0}^{\infty} a_n s^n$, with $a_n = 
	\sum_{m=1}^{n}
	\;
	\sum_{\substack{
			k_1 + k_2 + \cdots + k_m = n\\
			1 \le k_j \le K-1
	}}
	q_{k_1}\, q_{k_2}\, \cdots\, q_{k_m}
	$, for $n>0$. This function involves only the probabilities $p_0,p_1,\ldots p_K$ that determine the distribution of $Y$. Finally,  we get that $\pi_n=a_n\pi_0$, for all $n>0$.
\end{proof}
	\begin{example} 
	For the particular case $ K = 2 $, the process $ \{Z_n\}$ represents a random walk with transition probabilities 
$ p_0 = P_{i,i-1}$, $ p_2 = P_{i,i+1}$, and $ p_1 = 1 - (p_0 + p_2) $. 
	In this case, condition~\eqref{eq:cond} simplifies to $ p_2 < p_0 $, 
	which constitutes a well-known sufficient condition for the existence of a stationary distribution in a random walk. Note that the condition \eqref{eq:cond} ensures the positive recurrence of the Markov chain.
	\end{example}
Next we prove the asymptotic Normal distribution of the volume of water inside the reservoir under the Moran model for a dam with infinite capacity.
	
\begin{proposition}{\bf Stationary water balance (semi-infinite case)}\\
	Let $Z_n$ denote the volume of water stored in a reservoir with infinite capacity at time $n\ge 0$. Let $Y_n$ denote the total input registered during the period $(n,n+1)$. We assume that $\{Y_n, n\ge 0\}$ is a sequence of {i.i.d.} random variables with common mean $\mu_y$ satisfying condition \eqref{eq:cond}. Let $m_{c_0}=c_0\lim_{n\rightarrow \infty}\PP(Z_n \ge c_0)$. Then $\lim_{n \rightarrow \infty} \frac{Z_n}{n} \rightarrow \mu_y-m_{c_0}$.
\end{proposition}	
\begin{proof}	
The proof follows by the same arguments as Proposition \ref{prp:as}.	
\end{proof}

\bigskip

\begin{theorem}\label{theo:clt_infinite} {\bf A Central limit theorem for storage volume (semi-infinite case)}\\
	Let $Z_n$ denote the volume of water stored in a reservoir at time $n \ge 0$. Let $Y_n$ denote the total input registered during the period $(n,n+1)$. We assume that $\{Y_n, n\ge 0\}$ form a sequence of i.i.d. random variables satisfying condition \eqref{eq:cond} and with common mean $\mu_y$, and variance $\sigma_y^2 < \infty$. Let also assume the couple $(Y,Z)$ fits the model expressed in \eqref{eq:moran_infinite}. Then 
	\[
	\sqrt{n}\left(\frac{Z_n}{n} - \mu_g\right) \stackrel{\cal D}{\longrightarrow} N(0, \sigma_g^2) \ \ {\rm as \ } n \rightarrow +\infty,
	\]
	with $
	\sigma_g^2=\sum_{y,z} p_y \pi(z) \widetilde{g}(y,z)^2+2 \sum_{y> 0, z> 0}\sum_{y',z'} p_y\pi(z) h_0(y,z,y',z') \widetilde{g}(y,z)\widetilde{g}(y'z')$,
	where $\mu_g=\mu_y -(1-\pi_0)$; $\widetilde{g}(y,z)=y-1_{\{z>0\}}-\mu_g$; and $h_0(y,z,y',z')=\sum_{n=1}^{\infty}\PP(T_{0}\geq n, Y_n=y', Z_n=z'\mid Y_0=y,Z_0=z)$, with $T_{0}=\inf\{n:Y_n=0,Z_n=0\}$.

\end{theorem}

\begin{proof}
	To prove this result we consider the work on the central limit theorem (CLT) for ergodic Markov chains presented in \cite{Port1994} which is valid on countable state space. 
	
	Let ${\bf X}_0=(Z_0,0)$ and for $n>0$ let ${\bf X}_n=(Y_{n-1},Z_{n-1})$. Then, ${\bf X}=\{{\bf X}_n; n\ge 0\}$ is a Markov chain on a countable state space $E=E_y \times E_z$, with transition matrix ${\widetilde{\bf P}}$ and stationary distribution $\widetilde{\pi}$ defined in \eqref{eq:pi2}.

	We define a reward function $g(i,j):=i-1_{\{j>0\}}$, for all $(i,j)\in E$. Then we have that $\EE_{\widetilde{\pi}}[g(X_1,X_2)]=\mu_g$. Let us also define $\widetilde{g}(i,j):=g(i,j)-\mu_g$. Finally, we consider the cumulative sum $S_n:=\sum_{k=0}^n g(X_{1,k},X_{2,k})$.
	
	In our model, this cumulative reward at $n$ coincides with the second coordinate of the chain ${\bf X}$ but one step ahead, that is, we get that $S_n=Z_{n+1}$.
	This identity will allow us to analyse the behaviour of the volume stored in the dam, $\{Z_n, n\geq 0\}$ through the well-developed theory of additive functionals of Markov chains (\cite{Port1994},\cite{MeynTweedie09}).
	In particular, we use {\bf Theorem 1} in \cite{TrevezasLimnios09}, which is itself a restatement of the original result in \cite{Port1994}. We have adapted the statement in terms of our notation as follows.

	\noindent Assume that for some state $a=(a_1,a_1)\in E$, 
	\begin{equation}\label{eq:port1}
		{\sum_{y,z,y',z'}} \widetilde{\pi}(y,z)h_a(y,z,y',z') |\widetilde{g}(y,z)| |\widetilde{g}(y'.z')|< \infty
	\end{equation}
	and set  
	$\sigma_g^2=\sum_{y,z} \widetilde{\pi}(y,z) \widetilde{g}(y,z)^2+2\sum_{(y,z)\neq a}\sum_{y',z'} \widetilde{\pi}(y,z)h_a(y,z,y',z') \widetilde{g}(y,z)\widetilde{g}(y',z'),$
	where $\widetilde{g}$ and $h_a$ are defined above.  Then, for any initial distribution,
	\begin{equation}\label{eq:port2}
		\sqrt{n}(\frac{S_n}{n}-\mu_g)\stackrel{D}{\rightarrow}N(0,\sigma_g^2).
	\end{equation}
	
	To get \eqref{eq:port2}, we need to check that  \eqref{eq:port1}. We take for instance $a=(0,0)$. From Proposition \ref{prp:pi_infinite}, we know that the two-dimensional chain ${\bf X}$ is ergodic. Consequently, the expected time of first entrance into state  $(0,0)$, denoted by $\EE[T_0]<\infty$, is finite. Let   $N_{(y',z')}^{0}$ denote  the number of visits to state $(y',z')$ before yielding state $(0,0)$ for the first time, then we have that $h_0(y,z,y',z')=\EE_{(y,z)}[N_{(y',z')}^{0}]<\infty$ for any $y,z,y',z'$. So, to prove \eqref{eq:port1}, it is enough that 
	\begin{equation}\label{eq:port3}
		\EE_{\widetilde{\pi}}[(\widetilde{g}(X_1,X_2))^2]<\infty
	\end{equation}
	with $\EE_{\widetilde{\pi}}[\cdot]$ denoting expectation under the stationary distribution. The second moment in \eqref{eq:port3} is guaranteed since $Var(Y)=\sigma_y^2< \infty$. 
	
	In our particular case we have that $S_n=Z_{n}$ and $\mu_g=\mu_y-(1-\pi_0)$, which can be deduced directly from Proposition \ref{prp:as}  with $M=0$, as in the case of infinite capacity there is no overflow.
\end{proof}
\section{Stochastic properties of a reservoir system based on the Moran model with continuous input}\label{sec:continuous}

Let $G(y) = \mathbb{P}(Y_n \leq y)$ be the cumulative distribution function of $Y_n$ for $y \in E_y \subset \R^+$, with $\PP(Y_n=0)\! > \!0$, implying that  $G$ has an atom at zero. Suppose the reservoir releases a fixed amount, $c_0 > 0$, whenever available, and has a maximum capacity of $C_1$. Under these conditions, the transition function is given by:

\begin{equation}\label{eq:moran_cont}
P(z_1,[0,z_2])=\left\{\begin{array}{ll}
	G(c_0-z_1), & 0 \leq z_1 \leq c_0, \ z_2=0 \\
	G(c_0-z_1+z_2), & \max\{0, c_0-z_1\} < z_2 < C_1-c_0 \\
	1-G(C_1-z_1), & 0 \leq z_1 \leq C_1 -c_0, \ z_2 =C_1-c_0 \\
	0, & {\rm elsewhere}.
\end{array}\right.
\end{equation}

\begin{remark}
\medspace
\begin{enumerate}
	\item In our context $P(Y_n=0)>0$ is a natural physical hypothesis. It follows that $G(0)\!>\!0$; hence $G$ has an atom at zero and is therefore not absolutely continuous with respect to the Lebesgue measure.
	
\item  With the definitions above, we have that
\begin{eqnarray*}
	&&\hspace{-3cm} \PP((Y_{n+1},Z_{n+1})\in B_y \times B_z\mid (Y_{n},Z_n)=(y,z))=\\
	&&\qquad =\PP(Y_n \in B_y)\PP(Z_{n+1}\in  B_z\mid Z_{n}=z,Y_n=y)\\
	&&\qquad={\tt 1}_{\{((z+y)\wedge (C_1-c_0))^+\in B_z\}} G(B_y)
\end{eqnarray*}
for any Borel set $B_y\times B_z \subseteq E_y \times E_z$.

\item If $\{Y_n\}$ are i.i.d., then $\{Z_n, n\ge 0\}$ is a MC with a general state space and transition function given in \eqref{eq:moran_cont}, otherwise, we have that $(Y,Z)$ is a two-dimensional MC.
\item If we denote $U_n$ the actual input, then, at any time $n$, we have that $U_n=\min\{Y_n, C_1-Z_n\}$.
\end{enumerate}
\end{remark}
\noindent For the next Proposition, we would need some regularity conditions on the Markov chain $Z$.  In particular we are interested in ergodicity properties, and the existence of an invariant measure $\pi$, which satisfies that $\int_{E_z} P(z, A)\pi({\rm d} z) =\pi(A)$ for any measurable set $A \subseteq E_z$, where $P(z,A)=\PP(Z_1 \in A \mid Z_0=z)$ is the transition kernel. We can use Theorem 13.0.1 on page 313 of \cite{MeynTweedie09} ({\it Ergodic Theorem}),  to establish that if $\{Z_n; n \geq 0\}$ is a Markov chain on a state space $E_z$ with transition kernel $P$, which is irreducible, aperiodic, and positive recurrent and with a unique stationary distribution $\pi$, then, for any initial state $z \in E_z$,
\begin{equation}\label{eq:doeblin1}
\lim_{n \to \infty} P_n(z, A) = \pi(A) \quad \text{for any measurable set } A \subseteq E_z.
\end{equation}

A sufficient condition for the above result to hold is the Doeblin condition, i.e. there exists $ \delta > 0 $ constant, $ n_0 \ge 1$, and a probability measure $\nu$ on $E_z$ such that for all $z\in E_z$ and all measurable set $ A \subseteq E_z$,
\begin{equation}\label{eq:doeblin2}
P_{n_0}(z, A) \geq \delta \nu(A).
\end{equation}
Under the Doeblin condition, there exists a unique stationary distribution $ \pi$ for the Markov chain, the chain is positive recurrent, and for any initial state $z \in E_z$, \eqref{eq:doeblin1} holds.

\begin{lemma}
Let $G(y) = \mathbb{P}(Y_n \leq y)$ be the cumulative distribution function of $Y_n$ for $y \in E_y \subset \R^+$, with $\PP(Y_n=0)\! > \!0$. The Markov chain $Z=\{Z_n; n\ge 0\}$, with state space $E_z=[0,C_1-c_0]$ and  with the kernel given in \eqref{eq:moran_cont} satisfies the Doeblin condition.
\end{lemma}
\begin{proof}
Let $k=\left[\frac{C_1-c_0}{c_0}\right]$ be the integer part of the division $(C_1-c_0)/c_0$, and take any interval $A=[0,z_1] \subseteq E_z$. Then, given that $\PP(Y_n=0)>0$, we can check that 
\[
P_{k+1}(z,[0,z_1]) \ge (\PP(Y_n=0))^{k}({G}(z_1+c_0)-G(c_0)).
\]
 Let us take $n_0 = k+1$ and $\delta = \mathbb{P}(Y_n = 0)^k$, and, for any $0 \le z_0 \le z_1 \le C_1 - c_0$, define $\nu([z_0, z_1]) := G(z_1 + c_0) - G(z_0 + c_0)$, 
 which is the measure induced by $G$ on $E_z = [0, C_1 - c_0]$. With this definition, we conclude that the Doeblin condition is satisfied for $Z$.

\end{proof}

Finally, as known, for any measure $\rho$ it can be defined the following operator
\[
\rho P_n(A)= \int_{E_z} P_n(x,A)\rho({\rm d} x)
\]
for any $A \subset E_z$.

If $\rho_0$ denotes the initial law of the Markov chain $Z$, then the occupation probability distribution at time $n$ is $p_n(\cdot)= \rho_0P_n(\cdot)$. In particular for intervals of type $[0, z]$, we are interested in $F_n(z):=\PP(Z_n \leq z)=p_n([0,z])$, for any $z \in R_+$, and its limit for $n \rightarrow +\infty$, i.e. $F(z)$.

\bigskip

\begin{proposition}\label{pr:2}
Let $Z=\{Z_n, n \geq 0\}$, $Y=\{Y_n, n\geq 0\}$ as defined above. For all $Y_n$ we denote $\bar{G}(y)=\PP(Y_n > y)$, for $y \in E_y \subseteq [0,+\infty)$. For $n \geq 0$, let us denote $F_n(z)=\PP(Z_n \leq z)$, for $z \in E_z$, and  $F(z)= \underset{n\rightarrow +\infty}{\lim}F_n(z)$, for $z\in E_z$. Then
$$F(z)= \displaystyle{\int_0^{C_1-c_0}\! \! G(c_0+z-x) \  F ({\rm d}x)}.$$

\end{proposition}

\begin{proof}
Let $z \geq 0$, then
\begin{eqnarray}\label{eq:F1}
	\nonumber&&\hspace{-1.5cm} F_{n+1}(z)=\PP(\max\{0,Z_n+Y_n-c_0\}\leq z)\\
	\nonumber &&=\PP(Z_n+Y_n-c_0\leq z)=\PP(Z_n+\min\{Y_n,C_1-Z_n\}\leq z+c_0)\\
	\nonumber &&=\PP(\min\{Z_n+Y_n,C_1\}\leq z+c_0)\\
	\nonumber &&=1-\PP(Z_n+Y_n > z+c_0,\ C_1>z+c_0)\\
	&&= 1-{\tt 1}_{\{C_1>z+c_0\}}\EE \left[\PP(Z_n+Y_n > z+c_0 \mid Z_n)\right]
\end{eqnarray}
We have that $\PP(Z_n+Y_n>c_0+x\mid Z_n=z)=1-G(c_0+x-z)$, 
then, we can substitute in \eqref{eq:F1} and obtain
\begin{eqnarray}\label{eq:F2}
	\nonumber&&\hspace{-1.5cm} F_{n+1}(z)=1-{\tt 1}_{\{C_1>c_0+z\}}\EE\left[\bar{G}(c_0+z-Z_n)\right]
=1- \int_0^{C_1-c_0}\bar{G}(c_0+z-x) F_n({\rm d}x)\\
	&&=\int_0^{C_1-c_0} F_n({\rm d}x) -\int_0^{C_1-c_0} \bar{G}(c_0+z-x)  F_n({\rm d}x).
\end{eqnarray}
Taking limits and using that $G$ is a continuous and bounded function, we get that
\[
\underset{n\rightarrow +\infty}{\lim} F_{n+1}(z)=F(z)=\int_{0}^{C_1-c_0}\! G(c_0+z-x) F({\rm d}x).
\]
\end{proof}

\section{Usual performance measures for reservoir management}\label{sec:measures}
Let us consider the more general case where $\{Y_n, n \ge 0\}$ is a Markov chain, and let denote $\widetilde{\bf P}$ the transition  probability matrix of the 2-dimensional MC $(Y,Z)$ whose state space $E=E_y \times E_z$ is considered as noted above considering the lexicographic order prioritizing $Z$. \\
\noindent In particular we are interested in the following dependability measures.
\subsection{Reliability measures}
We call ``empty'' the state in which the volume of water in the reservoir is not enough to provide one single year of water demand. 
Let $Z_0=z_0$ denote the initial volume of water in the reservoir, with $z_0 \in E_z$.
\begin{itemize}
	\item {\bf Probability of first emptiness}\\
	Let $R_{z_0}(n)$ denote the probability that the reservoir is not empty at any time before $n$, given the initial volume of $z_0$. For $n > 0$, 
	\begin{equation}\label{eq:R}
		R_{z_0}(n)= \sum_{y \in E_Y} \PP(Z_k\neq 0, k=1,\ldots, n \mid Z_0=z_0,Y_0=y) \PP(Y_0=y) 
	\end{equation}
	Let  $\widetilde{P}_0$ denote the submatrix of $\widetilde{P}$ restricted to transitions between states in $E_0=E_y \times E_{z,0}$, with $E_{z,0}=E_z \setminus \{\bf 0\}$, with $\{\bf 0\}$ representing the  empty-state class. Then
	$R_{z_0}(n)= ({\bf e}_{z_0}\otimes{\bf \pi}_y )\widetilde{{\bf P}}_0^n \ {\bf 1}, $
	where ${\bf \pi}_y$ is the stationary distribution of $Y$. We denote ${\bf e}_{z_0}$ a vector of all components equal to 0 except for a 1 at the index of $z_0 \in E_{z,0}$, and, ${\bf 1}$ is an all-ones vector of the adequate dimension to produce the sum.
	\bigskip
	
	\item {\bf Probability of emptiness (not necessarily the first)}\\
	Let $A_{z_0}(n)$ denote the probability that the reservoir is not empty at time  $n$, given the initial volume of $z_0$. For $n > 0$
	\begin{equation}\label{eq:A}
	A_{z_0}(n)=({\bf e}_{z_0}\otimes{\bf \pi}_y ) \widetilde{{\bf P}}^n\ {\bf 1}, 
	\end{equation}
	for $n >1$.
	For this case, we denote ${\bf e}_{z_0}$  as a vector of zeros with a 1 in the position corresponding to the index of $z_0$ in $E_z$.
	
	\item {\bf Mean time to emptiness}\\
	Starting at conditions $Z_0=z_0$, the mean time to first reach the empty state  can be  approximated by 
	\begin{equation}\label{eq:mttf}
	MTTE_{z_0}= \sum_{n >0} n (R_{z_0}(n\!-\!1)-R_{z_0}(n))=({\bf e}_{z_0}\otimes{\bf \pi}_y )({\bf I}- \widetilde{{\bf P}}_0)^{-1}\ {\bf 1}.
	\end{equation}

	\item {\bf Long run expected value of water loss per unit time }\\
	Let us denote $C_O(n)=\sum_{k=0}^n \max \{ (Z_k+Y_k-C)>0\}$, then the expected value of water loss per unit time can be determined 
	\begin{equation}\label{eq:loss}
	\underset{n \rightarrow \infty}{\lim}\frac{ C_O(n)}{n}=\sum_{y,z}\widetilde{\pi}(y,z)(z+y-C)^+
	\end{equation}
	where $\widetilde{\pi}(y,z)$ is the stationary distribution of the 2-dimensional process (see Lemma \ref{lemma:pi}).

	\item {\bf Safety level measure}\\
	When the reservoir level drops below a critical threshold, cavitation may occur, causing damage to hydraulic structures and operational failure. To prevent this, a minimum storage level $z_{\text{safe}}$ is defined. The safety level measure $S_{\text{safe}}$ quantifies the probability that the reservoir volume remains above this limit, and is defined at time $n$ by
	\begin{equation}\label{eq:S}
	S_{z_0}(n) = \mathbb{P}(Z_n \ge z_{\text{safe}} \mid Z_0 = z_0).
	\end{equation}
	The long term measure can be defined as $S = \sum_{z \ge z_{\text{safe}}} \widetilde{\pi}(z)$.
	
\end{itemize}
\subsection{Resilience measures}\label{sec:resilience}
In this section we introduce two different resilience concepts useful for  dam management. They describe how long the reservoir can operate safely and how quickly it recovers from extreme conditions. 
First, we let us recall the concept of resilience in a general context.
\begin{definition}{Resilient state}(\cite{Tan2023})\\
	Consider a system with $K$ performance levels evolving as an irreducible Markov chain $\{Z_n, n \ge 0\}$, where $Z_n$ denotes the system state at time $t_n$. The state space $E = \{1, \ldots, K\}$ is ordered by decreasing performance and partitioned as follows:
	\begin{itemize}
		\item Perfect operating states: $E_1 = \{1, \ldots, k_1\}$;
		\item Imperfect operating states: $E_2 = \{k_1+1, \ldots, k_2\}$;
		\item Interrupted (non-resilient) states: $E_3 = \{k_2+1, \ldots, K\}$.
	\end{itemize}
\end{definition}

States in $E_1 \cup E_2$ are considered \emph{resilient}, while those in $E_3$ represent failure or interruption.  
The transition matrix ${\bf P}$ can thus be partitioned into submatrices ${\bf P}_{E_iE_j}$ describing transitions from $E_i$ to $E_j$.  
We define the following first-passage times: $
T_1 = \inf\{n: Z_n \in E_1\}$, and, $T_3 = \inf\{n: Z_n \in E_3\}$.

Following \cite{Tan2023}, resilience encompasses two complementary capabilities:

\begin{itemize}
	\item \textit{Resistant resilience}: the ability to maintain performance under disruption;
	\item \textit{Recovery resilience}: the ability to restore performance after disruption.
\end{itemize}

According to these measures, we consider in this paper the following resilience metrics:
	\begin{itemize}
		\item \textbf{Resistant Resilience:}  
		The expected number of times the system remains in perfect operating conditions before first entering a non-resilient state:
		\begin{equation}\label{eq:res1}
			\mathcal{R}_{\text{res}}
			= \mathbb{E}_i\!\left[\sum_{n=1}^{\infty} \mathbf{1}_{\{T_3 \ge n,\, Z_n \in E_1\}}\right]
			= \mathbb{E}_i[N_3^1],
		\end{equation}
		where $N_3^1$ counts visits to $E_1$ prior to failure.  
		This measure represents the \emph{endurance} or robustness of the system against disturbances.
		
		\item \textbf{Recovery Resilience:}  
		The expected number of times the system remains in a non-resilient state before returning to perfect operation:
		\begin{equation}\label{eq:res3}
			\mathcal{R}_{\text{rec}}
			= \mathbb{E}_i\!\left[\sum_{n=1}^{\infty} \mathbf{1}_{\{T_1 \ge n,\, Z_n \in E_3\}}\right]
			= \mathbb{E}_i[N_1^3].
		\end{equation}
		This quantity reflects the \emph{recovery capacity} of the system and is inversely related to expected restoration time.
	\end{itemize}

While reliability indicators describe the (long-term) dependability of reservoir operation, resilience metrics capture its short-term adaptive behaviour under disruptions and recovery phases.

\subsection*{Resilience metric calculation}
\noindent Let $E$ denote the finite state space of the Markov chain $\{Z_n, n\ge0\}$, with $|E|=k$ and let $A\subset E$, with $|A|=k'$, $k'<k$, be a target  class of states.\\ 
The first entrance time into $A$ is defined as $T_A := \inf\{n\ge1:\; Z_n \in A\}$, 
that is, the first visit to $A$ after time 0.

For any state $j\in E$, we define the random variable $N_j^A := \sum_{n=1}^{T_A} \mathbf{1}_{\{Z_n = j\}}$, representing the number of visits to state $j$ up to (and including) the hitting time of $A$.  
Its expectation, $g_A(i,j) := \mathbb{E}_i[N_j^A]$, can be computed using the fundamental matrix of the sub-chain restricted to $E\setminus A$ \cite{TrevezasLimnios09}:
\begin{equation}\label{eq:NjC_correct}
	g_A(i,j)=
	\begin{cases}
		{\bf e}_{k,i}^{\top}\,{\bf P}_{Ac}^0\,({\bf I} - {\bf P}_{ A})^{-1}\,{\bf e}_{k',j}, & j\notin A,\\[4pt]
		\mathbb{P}_i(Z_{T_A}=j), & j\in A,
	\end{cases}
\end{equation}
where: ${\bf P}_{A}^0$ is the transition submatrix from any state in $E$ to any state in $E\setminus A$; ${\bf P}_{A}$ is the transition matrix between states in ${E\setminus A}$; 
${\bf I}$ is the identity matrix of appropriate dimension; and, for a finite set $A$ with $|A|=m$ and any $i \in A$, ${\bf e}_{m,i}$ denotes the $i$th standard basis in $\R^m$. 

For any disjoint subsets $A, B \subset E$, the expected number of visits to all states in $B$ before first entering $A$ is $g_B(i, A) := \mathbb{E}_i[N_A^B] = \sum_{j\in A} g_B(i,j)$. Accordingly, the two resilience metrics are obtained as: $\mathcal{R}_{\text{res}} = g_{E_3}(i, E_1)$, and $\mathcal{R}_{\text{rec}} = g_{E_1}(i, E_3)$, where $E_1$ denotes the set of perfect operating states and $E_3$ the set of non-resilient states.  
Hence, $\mathcal{R}_{\text{res}}$ quantifies the expected number of visits to perfect operating states before failure, while $\mathcal{R}_{\text{rec}}$ measures the expected number of visits to failed states before full recovery. In other words,  $\mathcal{R}_{\text{res}}$ characterizes the reservoir's endurance against extreme hydrological conditions, whereas $\mathcal{R}_{\text{rec}}$ quantifies its ability to recover normal operation after a disruption.


\section{A real case-study: The Quiebrajano reservoir}\label{sec:quiebrajano}

The Quiebrajano reservoir, located in the province of Jaén (southern Spain), has been in operation since 1976 and constitutes a key component of the regional water supply system. The reservoir has a catchment area of approximately $A = 99~\text{km}^2$ and a total storage capacity of $C_1 = 32~\text{hm}^3$. For the purposes of this study, the average annual water withdrawal for urban supply estimated to be $c_0 = 10~\text{hm}^3$. In this study, observed inflow and storage data are considered on a yearly basis from 1~October~1999 to 1~May~2025, as provided by the Guadalquivir River Basin Authority (\url{https://www.chguadalquivir.es/saih/}). 

\subsection{The Moran model for the Quiebrajano case}
On the website that provides the data, SAIH (Automatic Hydrological Information System),  the available data correspond to flow rates expressed in cubic meters per second (\( \mathrm{m^3/s} \)). These values represent the average inflow rate for each month. To estimate the total volume of water entering the reservoir during a given month, the units can be converted from \( \mathrm{m^3/s} \) to cubic hectometres (\( \mathrm{hm^3} \)) by multiplying by \( 60 \times 60 \times 24 \times 30 \times 10^{-6} \), assuming each month has 30 days. This factor accounts for the number of seconds in a 30-day month and converts cubic meters into cubic hectometres. Summing the resulting monthly volumes provides the total annual inflow in \( \mathrm{hm^3} \). 
In this study, the analysis is conducted using annual inflow data rather than monthly values. This choice is motivated by the fact that the present work constitutes an initial exploration of the available data, focusing on long-term storage behaviour rather than short-term seasonal dynamics. Consequently, the periodicity that would arise from a monthly formulation is not explicitly modelled at this stage.\\
Data for the Quiebrajano area are available from 1 October 1999, and we downloaded information up to 1 October 2024. Therefore, the sample size is $N=26$. We denote $Z_0$ the initial volume of water stored, that is, at the date 01/10/1999, and for $n\ge 1$, $Z_n=\min\{\max\{Z_{n-1}+Y_{n-1}-c_0,0\},C_1\}$ 
is the volume of water stored at the end of the $n$ year, that can take values in the interval $[0,C_1-c_0]$. The SAIH website also provides information on the annual outflow of water from the reservoir. In this study, we consider the outflow as a deterministic quantity, estimated as the mean annual volume released during the 26-year observation period, which amounts to approximately $C_0=10 \ hm^3$. According to official sources, the maximum capacity of the Quiebrajano reservoir is approximately 31.6 to 32 $hm^3$. This corresponds to its total storage capacity, while the useful capacity may be slightly lower depending on operational conditions. In this study we will consider $C_1=32$.

The inflow into the dam during an interval of time, say, 1 year, is a continuous random variable on the interval $[0,+\infty)$.  We will proceed as usual to make a simplifying  approximation in order to work with discrete quantities instead of with  continuous ones (\cite{Lochert1979}). We will take as unit of discretization $\delta=c_0$ and consider for the interval $[0,C_1-c_0]$ a similar discretization scheme as the one suggested in \cite{Moran59}. 
Let us consider $k$ such that $\left(k-\frac{1}{2}\right)c_0 < C_1-c_0 <  \left(k+\frac{1}{2}\right)c_0$, and define the following intervals: $I_0=[0,\frac{1}{2}c_0]$,$I_j=((j-\frac{1}{2})c_0,(j+\frac{1}{2})c_0]$ for $j=1,3,5,\ldots, k$, $I_k'=(( k-\frac{1}{2})c_0,C_1-c_0]$, and $I_{k+1}''=((k+\frac{1}{2})c_0, +\infty)$. 
For  the values of $Y_n$ we consider  discretization according to the set $E_y=\{I_0,\ldots, I_k, I_{k+1}''\}$ and for the values of $Z_n$ we consider $E_z=\{I_0,\ldots, I_{k-1}, I_{k}'\}$. For the Quiebrajano case we have that $C_1=32$, $c_0=10$,  and $k=3$, then we have $k+1=4$ states in $E_z$ and $k+2=5$ states in $E_y$. 
No statistical evidence against the independence of the variables $Y_n$ was found after applying Fisher's independence test. So, based on $N=26$ yearly records, we have estimated the distribution probability of  $Y$, the total inflow in a year,  which is given by

\begin{table}[ht]
\centering
\begin{tabular}{|rrrrr|}
	\hline
	$I_0$ & $I_1$ & $I_2$ & $I_3$ & $I_4''$  \\ 
	\hline
0.3462 & 0.3846 & 0.1538 & 0.0385 & 0.0769   \\ 
	\hline
\end{tabular}
\end{table}

We use the Moran model in the finite capacity case, discussed previously (see Section \ref{sec:finite}), and describe the dynamic evolution of the volume of water in the Quiebrajano dam using the  Markov chain $Z$, defined on the state space $ E_z=\{I_0,I_1,I_2,I'_3\}$, which has a total of 6 states. The corresponding transition matrix is given by
\[
\mathbf{P}_z =
\begin{bmatrix}
 0.7308 & 0.1538 & 0.0385 & 0.0769 \\ 
 0.3462 & 0.3846 & 0.1538 & 0.1154 \\ 
 0.0000 & 0.3462 & 0.3846 & 0.2692 \\ 
 0.0000 & 0.0000 & 0.3462 & 0.6538
\end{bmatrix}
\]
We can calculate the stationary distribution of $Z$ directly from this matrix, or alternatively, obtain it from the estimated inflow distribution using the algorithm described in Section \ref{algorithm},
\begin{table}[ht]
	\centering
\begin{tabular}{|rrrr|}
	\hline
	$I_0$ & $I_1$ & $I_2$ & $I_3'$  \\ 
	\hline
0.2547 & 0.1980 & 0.2389 & 0.3084  \\ 
	\hline
\end{tabular}
\end{table}
which, among other results, allows us to estimate that the long-run probability of the dam being non-empty is below 75\%, consistent with the expectations for this reservoir according to official records.
\subsection{Analysis of stored volume}
Reliability and availability measures (defined in \eqref{eq:R} and \eqref{eq:A}, respectively) provide quantitative indicators that summarize the probability of maintaining adequate storage levels over time. They help reservoir managers assess the risk of water shortage, identify safe operating conditions, and define early-warning thresholds for drought events. 
From Figure \ref{fig:R} we see that, for the Quiebrajano reservoir, these reliability estimates suggest that starting from an intermediate storage level ($Z_0=I_1$), the probability of maintaining a nonempty condition drops from 0.88 in the first year to 0.66 after four years. This pattern reveals a moderate but steady increase in the risk of depletion, emphasizing the reservoir  vulnerability under sustained low inflow conditions. In operational terms, such results can support the Guadalquivir River Basin Authority in defining drought alert thresholds, adjusting annual withdrawal volumes, and planning supplementary inflows or demand restrictions to ensure long-term supply reliability.
\begin{figure}[ht]
\centering
\includegraphics[width=0.6\textwidth]{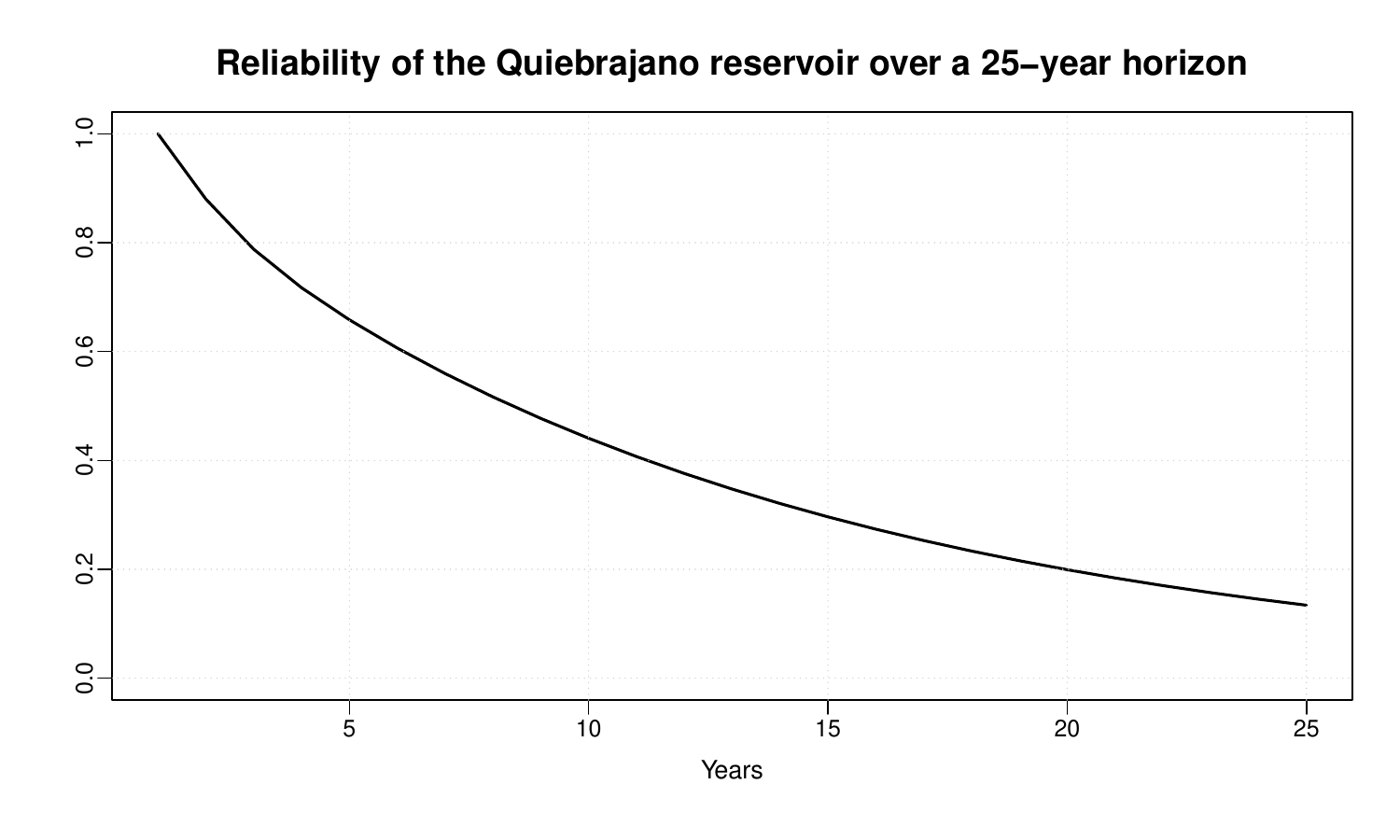}
\caption{Probability that the Quiebrajano reservoir remains non-empty throughout a 25-year period. The initial conditions are set to the ones registered in the dataset, i.e. the initial volume of stored water is at level $I_1=(1/2,3/2]$.} 
\label{fig:R}
\end{figure}
In Figure \ref{fig:A} we estimate the availability function defined in \eqref{eq:A}. This is the probability that the reservoir is non-empty (i.e., operational) at time $n$, given an initial storage level $Z_0=I_1$. Unlike the reliability function, which requires the reservoir to remain non-empty throughout the entire period, the availability only concerns the state at the specific time $n$, regardless of any empty periods that may have occurred earlier. From Figure \ref{fig:A} we observe that the long-run availability is obtained  74.5\%, meaning that, in steady operation, unless policy or system parameters change, the reservoir is expected to be non-empty about three quarters of the time.
According to the World Meteorological Organization \cite{wmo09}, a ``75~per~cent dependable yield'' indicates that the required water supply is available in at least three out of every four years. This convention supports interpreting long-term availability levels around 0.75 as satisfactory for reliable reservoir performance, whereas lower values would indicate increasing risk of supply failure.

\begin{figure}[ht]
\centering
\includegraphics[width=0.6\textwidth]{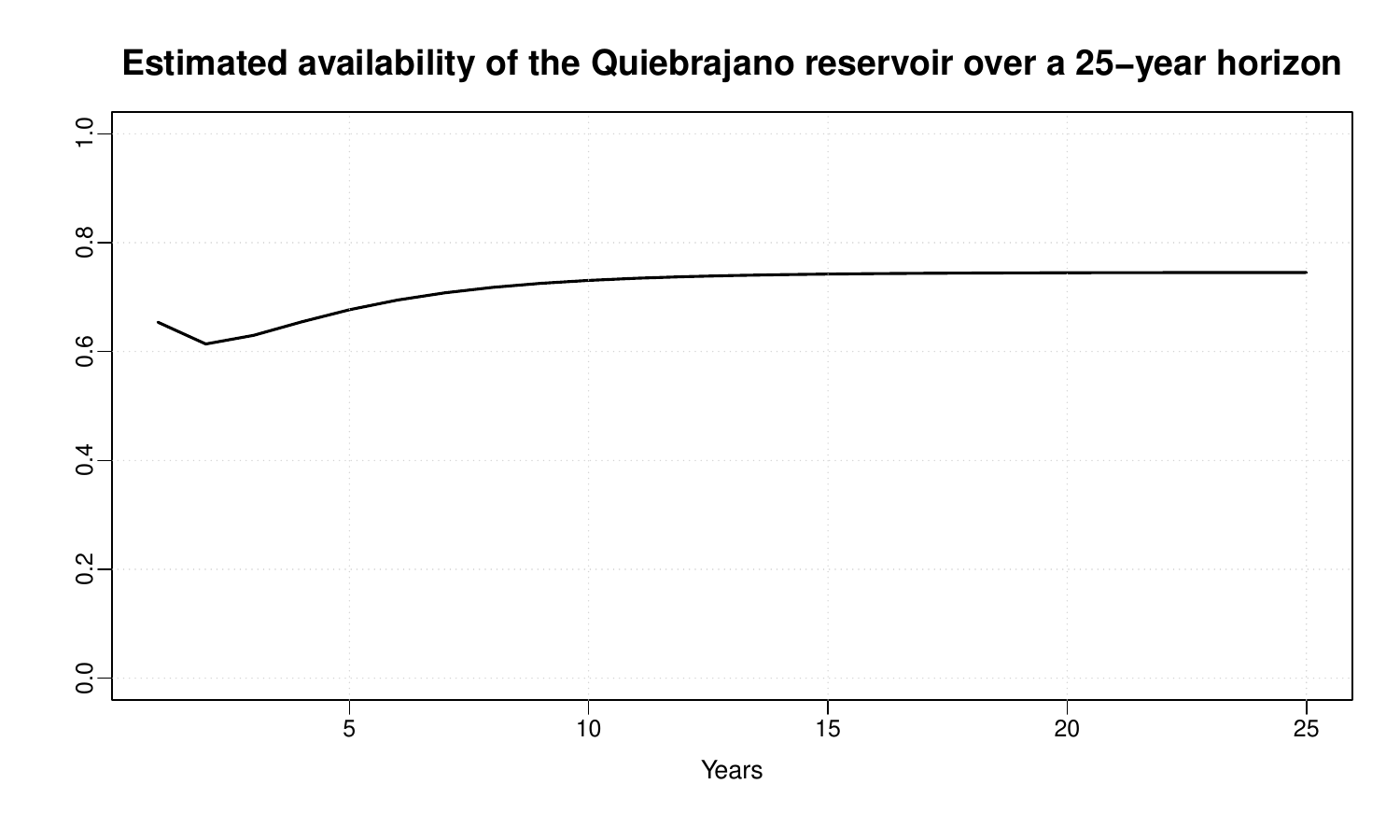}
\caption{ Probability that the Quiebrajano reservoir is empty (not necessarily for the first time) at year $n$ regardless of previous empty periods, over a 25-year simulation period. The initial volume of stored water is at level $I_1=(1/2,3/2]$, that is, the initial conditions are set to the ones registered in the dataset. 
 }\label{fig:A}
\end{figure}
Finally, Figure \ref{fig:expected} presents the historical series of observed storage volumes together with the expected values obtained from the Markov-based model. The grey dots represent the observed volumes, while the solid line shows the expected storage estimated from the fitted Moran model. The estimated curve provides a smooth approximation of the observed series, capturing the long-term equilibrium level of the reservoir. The convergence of the model expected value to the empirical mean of the data  (15.0526 $hm^3$) confirms that the model adequately reproduces the stationary behaviour of the storage process.
\begin{figure}[htb]
	\centering
	\includegraphics[width=0.6\textwidth]{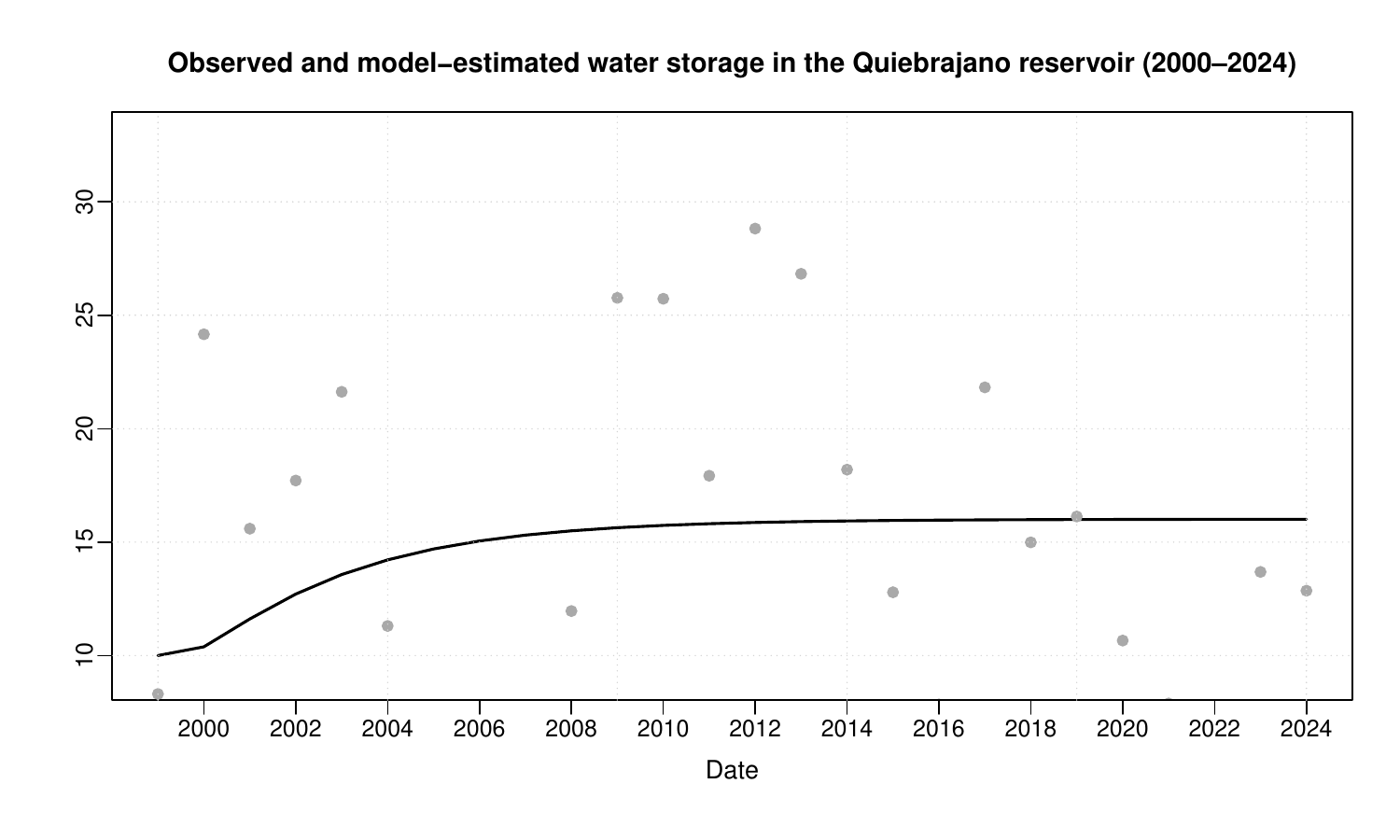}
	\caption{Observed storage volumes (grey dots) and expected values estimated from the Markov chain model (solid line). The model reproduces the long-term mean behaviour of the system, with the expected value stabilizing near the empirical average of the observed data, that is, $\overline{V}= 15.0526\  hm^3$.
	}\label{fig:expected}
\end{figure}
	
\subsection{Performance and resilience metrics}
	We consider only two classes of states:  $E_1=\{I_{0},I_1\}$ are the resilient states, while  $E_3=\{I_0,I_3\}$ are the non-resilient states. 
According to this classification we can formulate two measures of resilience as defined in Section \ref{sec:resilience}, i.e. resistance resilience ${\cal R}_{res}= \EE_{I_1}\left[N_{E_1}^{E_3}\right]$; and, recovery resilience  ${\cal R}_{rec}= \EE_{I_1}\left[N_{E_3}^{E_1}\right]$. Following Section \ref{sec:measures}, we calculate: $
{\cal R}_{res}=\sum_{j\in E_1}{\bf e}_{4,I_1}^{\top}{\bf P}_{E_3}^0\left({\bf I}-{\bf P}_{E_3}\right)^{-1}{\bf e}_{2,j}$; and, ${\cal R}_{rec}=\sum_{j\in E_3}{\bf e}_{4,I_1}^{\top}{\bf P}_{E_1}^0\left({\bf I}-{\bf P}_{E_1}\right)^{-1}{\bf e}_{2,j}$.

The estimated resistance resilience value is ${\cal R}_{res}=1.36$ years, which means that, on average, starting from a water level $I_1$, the Quiebrajano reservoir can remain within a safe operating range for about a year and four months before reaching a critical condition, either too empty or too full. This duration reflects the system  capacity to resist hydrological disturbances, such as droughts or large inflows, before entering a risky state. The estimated value is consistent with the reservoir management practices and its pronounced year-to-year variability in inflows, indicating a moderate but not high level of resilience over the long term.
On the other hand we have estimated ${\cal R}_{rec}=1.90$, which means that, once the reservoir reaches a critical state (either very low or very high storage), it takes on average almost two years to return to a normal, resilient operating level. In other words, this estimation suggests that the reservoir is resilient but with slow recovery, typical for semi-arid systems where hydrological replenishment is strongly seasonal and constrained, as it is the area in southern Spain, where the Quiebrajano dam is located.

\section{Conclusions and future research} \label{sec:conclusions}
This paper has developed a unified Markov-based framework for modelling reservoir behaviour and evaluating performance indicators such as reliability and resilience. The approach encompasses both finite- and infinite-capacity formulations, incorporates independent and Markov-dependent inflows, and extends naturally to continuous-state systems. Within this general setting, stationary water balance relations and asymptotic results, including central limit theorems for long-term storage behaviour, were established. The framework was then applied to the Quiebrajano dam to demonstrate its practical relevance.

Future work may progress along several complementary directions. A central aspect will be the close collaboration with domain experts. Some of the results from this study have already been presented to the engineer responsible for the dam, who not only considered the approach useful for operational support but also noted that our analysis revealed inconsistencies in the SAIH database that will require expert verification. This underscores the value of combining modelling efforts with practitioner insight, and we are currently exploring opportunities for continued cooperation.

From a methodological perspective, future developments include the treatment of non-stationary inflows to capture long-term climatic variability, the incorporation of non-deterministic outflows, and the integration of sensor-based monitoring data to obtain more detailed information and refine model calibration. Another promising avenue is linking the stochastic modelling framework with optimisation or decision-support procedures for adaptive reservoir management under competing objectives. Finally, analysing resilience using continuous monitoring data could offer valuable insights for decision-making and contribute to more robust and sustainable operational strategies.

\section*{Acknowledgments}
\noindent 

The authors gratefully acknowledge support from the Spanish Ministry of Science and Innovation - State Research Agency through grant PID2024-156234NB-C22. This work is also supported in part by the IMAG-Maria de Maeztu grant CEX2020-001105-M / AEI / 10.13039/501100011033.

\noindent We gratefully acknowledge the support and valuable insights provided by Estrella Montoro Cazorla, Operations Manager of the Quiebrajano Dam (Confederaci\'on Hidrogr\'afica del Guadalquivir).

\noindent The hydrological and operational data used in this study were obtained from \url{https://saih.chguadalquivir.es} (last accessed on 1 November 2025).

\end{document}